\documentclass[preprint,11pt]{elsarticle}
\usepackage[a4paper,margin=25mm]{geometry}
\usepackage[T1]{fontenc}
\usepackage[utf8]{inputenc}
\usepackage{lmodern}
\usepackage{amsmath,amssymb,bm,mathtools,slashed}
\usepackage{graphicx,booktabs,amsthm,microtype,etoolbox}
\usepackage[hidelinks]{hyperref}
\biboptions{numbers,sort&compress}
\journal{The European Physical Journal Plus}
\let\textbf\textnormal
\let\textit\textnormal
\let\emph\textnormal
\myfooter[C]{\thepage}
\hypersetup{pdftitle={Wormhole-induced gravi-axion mass distributions and the second-harmonic crossover},pdfauthor={Amin Rezaei Akbarieh}}
\newtheoremstyle{upright}{6pt}{6pt}{\normalfont}{}{\normalfont}{.}{.5em}{}
\theoremstyle{upright}
\newtheorem{proposition}{Proposition}
\newtheorem{corollary}{Corollary}
\makeatletter
\renewcommand\section{\@startsection{section}{1}{\z@}{-3.5ex plus -1ex minus -.2ex}{2.3ex plus .2ex}{\normalfont\large}}
\renewcommand\subsection{\@startsection{subsection}{2}{\z@}{-3.25ex plus -1ex minus -.2ex}{1.5ex plus .2ex}{\normalfont\normalsize}}
\renewcommand\subsubsection{\@startsection{subsubsection}{3}{\z@}{-3.25ex plus -1ex minus -.2ex}{1.5ex plus .2ex}{\normalfont\normalsize}}
\patchcmd{\@address}{\itshape}{\normalfont}{}{}
\patchcmd{\pprintMaketitle}{\itshape}{\normalfont}{}{}
\patchcmd{\MaketitleBox}{\itshape}{\normalfont}{}{}
\makeatother
\myfooterfont{\footnotesize\normalfont}
\renewenvironment{proof}[1][Proof]{\par\noindent #1.\ }{\hfill$\square$\par}
\newcommand{\mpl}{M_{\rm P}}
\newcommand{\dd}{\mathrm{d}}
\newcommand{\ii}{\mathrm{i}}
\newcommand{\ee}{\mathrm{e}}
\newcommand{\RtildeR}{R\widetilde R}
\newcommand{\essinf}{\operatorname*{ess\,inf}}

\newcommand{\cO}{\mathcal O}
\newcommand{\cI}{\mathcal I}
\newcommand{\cJ}{\mathcal J}
\newcommand{\cD}{\mathcal D}
\newcommand{\cM}{\mathcal M}
\newcommand{\cG}{\mathcal G}
\newcommand{\cA}{\mathcal A}
\newcommand{\cC}{\mathcal C}

\newcommand{\cQ}{\mathcal Q}

\begin{document}
\begin{frontmatter}
\title{Wormhole-induced gravi-axion mass distributions and the second-harmonic crossover}
\author{Amin Rezaei Akbarieh}
\address{Department of Physics, Kocaeli University, 41001 Izmit, T\"urkiye}
\ead{amin.rezaeiakbarieh@kocaeli.edu.tr}
\begin{abstract}
We compute the distribution of the wormhole-induced gravi-axion curvature mass across baby-universe sectors when the first two harmonics are retained. For independent complex Gaussian coefficients, with an optional displacement of the first channel, an exact global-minimum condition gives the mass distribution and its threshold coefficient. At fixed nonzero second-harmonic width, the cumulative probability below mass $M$ scales as $M^5$. For a weak second harmonic, the unit-harmonic $M^4$ behavior is recovered in an intermediate low-mass range above the crossover scale. We also distinguish the actual relative mass-squared error from the probability that a phase-independent bound fails to certify the unit-harmonic approximation. In the common-background quadratic correlator with unit residues, the lower threshold remains zero and the two-harmonic tail produces an $s^{3/2}$ nonanalytic term in the Euclidean momentum squared $s$. The results are conditional on the chosen state and effective wormhole coefficients.
\end{abstract}
\begin{keyword}
Axion wormholes \sep Baby universes \sep Curvature mass \sep Mass distributions \sep Higher harmonics
\end{keyword}
\end{frontmatter}
\section{Introduction}
\label{sec:introduction}
Euclidean wormholes provide a proposed gravitational source of axion shift-symmetry breaking \cite{GiddingsStrominger1988,KalloshLindeLindeSusskind1995,HebeckerReview2018,AlonsoUrbano2019}. Their relevance to an axion mass depends on more than the existence of a classical saddle. A connected wormhole inserts a pair of charged operators, whereas a local symmetry-breaking coefficient is specified only after conditioning on the baby-universe sector or choosing a state on the corresponding operator algebra \cite{Coleman1988CC,MarolfMaxfield2020}. The distinction between a conditional mass and a state average is therefore part of the physical definition of the observable.

We consider a periodic pseudoscalar $\phi=\mu\theta$, with $\theta\sim\theta+2\pi$, coupled to the gravitational Pontryagin density
\begin{equation}
R\widetilde R\equiv\frac12\varepsilon^{\mu\nu\rho\sigma}R^{\alpha}{}_{\beta\mu\nu}R^{\beta}{}_{\alpha\rho\sigma}.
\label{eq:pontryagin-definition}
\end{equation}
This interaction is the parity-odd coupling of dynamical Chern--Simons gravity \cite{JackiwPi2003,AlexanderYunes2009}. It can arise from a gravitational chiral anomaly after massive Peccei--Quinn charged fermions are integrated out \cite{DelbourgoSalam1972,Fujikawa1979,Fujikawa1980,AlexanderCreque2023}. Anomaly matching fixes its coefficient, but does not itself generate a perturbative potential for a constant axion in a topologically trivial background. Wormhole-induced potentials have consequently been used in models of gravi-axion dark matter and dark energy \cite{AlexanderEtAl2025}.

The use of $\alpha$ parameters and the dependence of local couplings on baby-universe states are established parts of the wormhole framework \cite{HebeckerReview2018,AlonsoUrbano2019}. A different assumption, namely that the baby-universe Hilbert space is effectively one-dimensional, does not produce the same ensemble interpretation \cite{McNamaraVafa2020}. We work conditionally within the commuting-operator description. We do not derive this description, or a preferred state, from a completion of quantum gravity. A Lorentzian thermodynamic prescription has also been proposed in which the symmetric point dominates certain Peccei--Quinn models \cite{Kawana2026}. Its oscillatory weighting is distinct from the normalized positive state measure used below.

In the quadratic approximation, a positive average of sector propagators has a lower spectral edge given by the essential infimum of the sector curvature masses squared. The weight near that edge depends on the state. We use centered and displaced Gaussian states realized in the oscillator description to compute this weight.

Our main result is the curvature-mass distribution at the global minimum of the two-harmonic potential. We compute the coefficient of its low-mass asymptote, determine the weak-hierarchy crossover, and compare the actual unit-harmonic error with a phase-independent bound. The matched-mean comparison for centered and displaced states provides a single-harmonic reference.

Section~\ref{sec:mouths} summarizes the saddle and the effective mouth coefficient. Section~\ref{sec:bilocal} distinguishes localization from state preparation. Section~\ref{sec:gap} defines the conditional mass and the quadratic spectral threshold. Section~\ref{sec:prepared-states} derives the Gaussian state measures and the two-harmonic mass distribution. Section~\ref{sec:conclusions} states the interpretation and remaining limitations.

\section{Gravi-axion wormhole mouth operators}
\label{sec:mouths}

\subsection{Minimal ultraviolet origin and anomaly matching}
\label{subsec:uv-anomaly}

A minimal weakly coupled progenitor contains a complex scalar $\Phi$ and a Dirac fermion with chiral Peccei--Quinn charges,
\begin{align}
 \mathcal L_{\rm UV}
 \supset{}&
 \nabla_\mu\Phi^\dagger\nabla^\mu\Phi
 -\lambda\left(\Phi^\dagger\Phi-\frac{\mu^2}{2}\right)^2
 +\overline\Psi\,\ii\slashed D\Psi
 \nonumber\\
 &-y\left(\overline\Psi_L\Phi\Psi_R+\mathrm{H.c.}\right).
 \label{eq:uv-minimal}
\end{align}
Here $\mathcal L_{\rm UV}$ denotes the ultraviolet Lagrangian density, $\lambda$ is the scalar self-coupling, $\mu$ is the Peccei--Quinn symmetry-breaking scale, $y$ is the Yukawa coupling, and $\Psi_L$ and $\Psi_R$ are the left- and right-handed components of $\Psi$, respectively. In the broken phase, one can obtain
\begin{equation}
 \Phi=\frac{\mu+\rho}{\sqrt2}\,\ee^{\ii\theta},
 \qquad
 \theta=\frac{\phi}{\mu},
 \label{eq:polar-field}
\end{equation}
where $\rho$ is the radial scalar mode, $\theta$ is the dimensionless angular field, and $\phi$ is the corresponding canonically normalized pseudoscalar field. The masses of the radial mode and the Dirac fermion are denoted by $M_\rho$ and $M_\Psi$, respectively, and are given by
\begin{equation}
 M_\rho^2=2\lambda\mu^2,
 \qquad
 M_\Psi=\frac{y\mu}{\sqrt2},
 \label{eq:uv-masses}
\end{equation}
and the Yukawa interaction becomes $-M_\Psi(1+\rho/\mu)\overline\Psi\ee^{\ii\gamma_5\theta}\Psi$.
A local chiral rotation removes this phase from the mass operator, yielding the Fujikawa Jacobian.  With the convention in Eq.~\eqref{eq:pontryagin-definition}, we write the resulting Lorentzian Wess--Zumino term as
\begin{equation}
 S_{\rm WZ}^{(L)}
 =
 \frac{\cA_{\rm grav}}{384\pi^2}
 \int \dd^4x\sqrt{-g}\,\theta\RtildeR
 =
 \int \dd^4x\sqrt{-g}\,
 \frac{\phi}{4f_{\rm CS}}\RtildeR,
 \label{eq:wz-term}
\end{equation}
where
\begin{equation}
 f_{\rm CS}=\frac{96\pi^2}{\cA_{\rm grav}}\,\mu,
 \label{eq:fcs-relation}
\end{equation}
and $\cA_{\rm grav}$ is the mixed Peccei--Quinn--gravitational anomaly coefficient, while $f_{\rm CS}$ denotes the effective dynamical Chern--Simons coupling scale. For the minimal charge assignment, in which left- and right-handed fermions differ by one unit of Peccei--Quinn charge, the gravitational anomaly coefficient is $\mathcal{A}_{\rm grav}=1$. The overall sign of Eq.~\eqref{eq:wz-term} depends on convention choices, such as the orientation and the definition of $\gamma_5$, but none of our conclusions depend on this sign. The parity-even terms in the fermion determinant renormalize the kinetic term and higher-derivative operators, but after canonical normalization they do not generate a nonderivative perturbative potential for $\theta$. At the same time, the radial mode need not remain frozen in the wormhole throat, since in the two-form frame its kinetic function scales as $(\mu+\rho)^{-2}$ and a nonzero flux can therefore source $\rho$. Here the two-form frame refers to the dual formulation of the angular field in terms of a two-form gauge field.  The constant-modulus saddle used below is therefore a controlled reference solution when the radial response and higher-curvature corrections are small; in a general ultraviolet completion, the fully backreacted profile changes the classical action and prefactor \cite{CheongParkShin2024,CatinariUrbano2025}.  Our state-dependence result will be expressed in terms of the resulting mouth coefficient and is independent of this model-dependent replacement.

\subsection{Two-form saddle and quantized flux}
\label{subsec:gs-saddle}

At fixed modulus, corresponding to $\rho=0$, the Euclidean two-form action is
\begin{align}
 S_E[B,g]
 &={}
 \int_{\mathcal M}
 \left[
 -\frac{\mpl^2}{16\pi}R\star1
 +\frac{1}{2\mu^2}H_3\wedge\star H_3
 \right]
 +S_{\rm GHY},
 \nonumber\\
 H_3&=\dd B_2.
 \label{eq:two-form-action}
\end{align}
Here $\mathcal M$ denotes the Euclidean manifold, $g$ its metric, $B_2$ the two-form gauge potential, $H_3$ its three-form field strength, $\star$ the Hodge dual, $R$ the Euclidean Ricci scalar, and $S_{\rm GHY}$ the Gibbons--Hawking--York boundary term, and we use the nonreduced Planck mass as $\mpl=G^{-1/2}$. Our three-form is integer normalized,
\begin{equation}
 \int_{S^3}H_3=n,
 \qquad n\in\mathbb Z,
 \label{eq:flux-quantization}
\end{equation}
where the integer $n$ labels the quantized three-form flux and the corresponding wormhole charge sector. Equivalently, the convention with periods $2\pi\mathbb Z$ is
$H_3^{(2\pi)}=2\pi H_3$.  The duality to the compact scalar and the associated normalization are reviewed in \ref{app:duality}.

For the $O(4)$ ansatz, we have
\begin{align}
 \dd s^2&=\dd\tau^2+a(\tau)^2\dd\Omega_3^2,
 \nonumber\\
 H_3&=\frac{n}{2\pi^2}\,\omega_3,
 \qquad
 \int_{S^3}\omega_3=2\pi^2,
 \label{eq:o4-ansatz}
\end{align}
where $\tau$ is the Euclidean radial coordinate and $a(\tau)$ is the radius of the three-sphere slices. Note that primes below denote derivatives with respect to $\tau$, and $\dd H_3=0$ and $\dd\star H_3=0$ hold identically. The independent Einstein equation gives
\begin{equation}
 (a')^2=1-\frac{L_n^4}{a^4},
 \qquad
 L_n^4=\frac{n^2}{3\pi^3\mpl^2\mu^2}.
 \label{eq:friedmann-wormhole}
\end{equation}
The two branches meet smoothly at $a=L_n$, where $a'=0$, forming a throat of radius $L_n$ that connects two asymptotically flat ends. Differentiating Eq.~\eqref{eq:friedmann-wormhole} gives
\begin{equation}
 a''=\frac{2L_n^4}{a^5},
 \qquad
 R=-\frac{6L_n^4}{a^6},
 \label{eq:wormhole-curvature}
\end{equation}
so the curvature is of order $L_n^{-2}$ at the throat. The leading Einstein--three-form action of the complete smooth two-ended saddle is
\begin{equation}
 S_{n,{\rm conn}}^{(0)}
 =
 \frac{3\pi^2}{4}\mpl^2L_n^2
 =
 \frac{\sqrt{3\pi}}{4}\,
 \frac{|n|\mpl}{\mu},
 \label{eq:connected-action}
\end{equation}
where the superscript $(0)$ denotes the leading classical saddle-point contribution, while the subscript ``conn'' refers to the connected two-ended geometry. The derivation, including the background-subtracted asymptotic Gibbons--Hawking--York term, is given in \ref{app:action}.  If the geometry is cut at the throat and treated as a half-wormhole with a separate boundary variational problem, the internal boundary term changes the numerical coefficient.  In the conventions of Ref.~\cite{AlonsoUrbano2019}, one can write
\begin{equation}
 S_{n,{\rm cut}}^{(0)}
 =
 \frac{\sqrt{3\pi}}{8}\,
 \frac{|n|\mpl}{\mu}
 \left(1-\frac{2}{\pi}\right).
 \label{eq:cut-action}
\end{equation}
Equations~\eqref{eq:connected-action} and \eqref{eq:cut-action} correspond to different boundary and gluing prescriptions; the full smooth action is not twice the cut action.  We use the connected saddle to define the bilocal coefficient and do not assume a unique factorization of its internal-boundary contribution into two independent mouth actions.

The semiclassical description of the wormhole geometry is reliable when the throat radius is large in Planck units and the classical action is correspondingly large,
\begin{equation}
 L_n\mpl
 =
 \frac{1}{(3\pi^3)^{1/4}}
 \left(\frac{|n|\mpl}{\mu}\right)^{1/2}
 \gg1,
 \qquad
 S_{n,{\rm conn}}^{(0)}\gg1.
 \label{eq:semiclassical-domain}
\end{equation}
For the constant-modulus ultraviolet model, however, these conditions alone are not sufficient, since the masses of the radial mode and the fermion must also be compared with the inverse throat scale:
\begin{align}
 M_\rho L_n
 &=
 \frac{\sqrt{2\lambda}}{(3\pi^3)^{1/4}}
 \left(\frac{|n|\mu}{\mpl}\right)^{1/2},
 \nonumber\\
 M_\Psi L_n
 &=
 \frac{y}{\sqrt2(3\pi^3)^{1/4}}
 \left(\frac{|n|\mu}{\mpl}\right)^{1/2}.
 \label{eq:throat-mass-ratios}
\end{align}
Thus, a wormhole can lie well within the semiclassical gravitational regime while still having $M_\rho L_n\lesssim1$ or $M_\Psi L_n\lesssim1$; in such cases, the radial mode cannot consistently be treated as frozen, and a large-mass heat-kernel expansion of the fermion determinant is not controlled in the throat region.

\subsection{Anomaly, fluctuations, and the mouth coefficient}
\label{subsec:prefactor}

The metric in Eq.~\eqref{eq:o4-ansatz} describes a Euclidean Friedmann geometry and is therefore conformally flat. In four dimensions, the Pontryagin density can be written entirely in terms of the Weyl tensor as
\begin{equation}
 \RtildeR
 =
 \frac12\varepsilon^{\mu\nu\rho\sigma}C^{\alpha}{}_{\beta\mu\nu}C^{\beta}{}_{\alpha\rho\sigma},
 \label{eq:pontryagin-weyl}
\end{equation}
where $C_{\mu\nu\rho\sigma}$ is the Weyl tensor; the index contraction follows the convention of Eq.~\eqref{eq:pontryagin-definition}. Since the Weyl tensor vanishes for a conformally flat geometry, the $O(4)$-symmetric wormhole background satisfies
\begin{equation}
 \left.\RtildeR\right|_{O(4)}=0.
 \label{eq:pontryagin-zero}
\end{equation}
As a result, the parity-odd interaction does not modify the three-form Bianchi identity in the first-order dual formulation. The same conclusion follows from varying the parity-odd interaction with respect to the metric. Within the $O(4)$ ansatz, the variation vanishes because no nonzero parity-odd symmetric rank-two tensor can be constructed solely from an $O(4)$-invariant metric and radial data. Consequently, the anomaly term does not modify either the $O(4)$-symmetric classical field equations or the connected wormhole action in Eq.~\eqref{eq:connected-action}. This conclusion, however, applies only to the spherical sector; parity-odd fluctuations, nonspherical saddle configurations, or parent geometries with a nonvanishing Pontryagin number may still acquire an anomaly-induced phase. The vanishing of the spherical contribution was noted in Ref.~\cite{AlexanderEtAl2025}.

Let $\Delta_n$ denote the determinant-corrected coefficient multiplying a connected charge-$n$ bilocal insertion.  Schematically,
\begin{equation}
 \Delta_n
 =
 \ee^{-S_{n,{\rm conn}}^{(0)}}
 \cJ_n
 \cC_n
 \frac{\det{}'\cM_{{\rm gh},n}}
 {\sqrt{\det{}'\cM_{{\rm bos},n}}}
 \frac{\det\cQ_{E,n}}{\det\cQ_{E,0}}.
 \label{eq:bilocal-prefactor}
\end{equation}
Here $\cJ_n$ contains collective-coordinate Jacobians, $\cM_{{\rm bos},n}$ is the constrained quadratic operator for the metric, two-form, and any retained radial fluctuations, $\cM_{{\rm gh},n}$ includes diffeomorphism ghosts and the reducible two-form ghost complex, and $\cC_n$ records the conformal-factor contour and any associated phase.  Primes remove collective and exact zero modes.  Quantum scalar--two-form duality also carries a local Euler-density normalization and, on a noncompact manifold, its boundary completion.  This local, scheme-dependent factor is absorbed into the renormalized coefficient $\Delta_n$ together with the counterterms.  The Euclidean fermion operator is represented schematically by
\begin{equation}
 \cQ_E
 =
 \slashed D_E
 +M_\Psi\ee^{\ii\gamma_{5,E}\theta},
 \label{eq:fermion-operator}
\end{equation}
where $\cQ_{E,n}$ and $\cQ_{E,0}$ in
Eq.~\eqref{eq:bilocal-prefactor} denote this operator evaluated on the
charge-$n$ wormhole saddle and the fluxless reference background,
respectively, with the understanding that the compact scalar is implemented through the first-order scalar--two-form functional integral rather than treated as a real classical throat profile.

The gravitational bulk index density of the massless kinetic Dirac operator underlying the anomaly vanishes on the spherical saddle because of Eq.~\eqref{eq:pontryagin-zero}.  Thus there is no bulk topological source for a net chiral index under standard asymptotically trivial boundary conditions.  This does not exclude paired or accidental zero modes, nor boundary spectral asymmetry.  If exact Grassmann zero modes remain unsaturated, the purely bosonic coefficient vanishes and the mouth operator requires fermion insertions.  When the bosonic vertex is allowed, the determinant changes its modulus and a single-mouth matrix element may carry a spectral phase.  In a charge-conjugate diagonal pair this phase cancels from the Hermitian bilocal coefficient, while relative phases between harmonics or entries of a non-diagonal kernel can remain physical.  A complex contour need not admit a positive Gaussian representation.

The fluctuation problem is sensitive to boundary conditions.  Analyses that impose fixed axion charge or fixed three-form flux have found no negative modes for the asymptotically flat saddle, whereas other choices led to different conclusions \cite{HertogTruijenVanRiet2019,LogesShiuSudhir2022,HertogEtAl2024,MarolfMissoni2025}.  We therefore retain the contour data in Eq.~\eqref{eq:bilocal-prefactor}.  Our later conclusions require only that a nonzero bosonic bilocal coefficient exists; they do not require an assumed closed-form value of the determinant.

For a positive diagonal kernel one may choose a complex number $\lambda_n$ of mass dimension four such that
\begin{equation}
 \lambda_n=|\lambda_n|\ee^{\ii\delta_n},
 \qquad
 |\lambda_n|^2=\Delta_n>0.
 \label{eq:lambda-definition}
\end{equation}
Only $|\lambda_n|$ is fixed by a strictly diagonal bilocal term; $\delta_n$ records the phase convention for a factorized mouth amplitude and becomes meaningful relative to other harmonics or non-diagonal entries.  The long-distance output of the saddle calculation is therefore encoded in the mouth operators
\begin{equation}
 \cO_n(x)=\lambda_n\ee^{\ii n\theta(x)},
 \qquad
 \cO_{-n}(x)=\cO_n(x)^\dagger,
 \label{eq:mouth-operators}
\end{equation}
rather than in a local cosine potential with a predetermined coefficient.

\section{Bilocal interaction and localization in \texorpdfstring{$\alpha$}{alpha} sectors}
\label{sec:bilocal}

\subsection{Dilute gas and the bilocal kernel}
\label{subsec:dilute-gas}

The spacetime-integrated charge-$n$ mouth insertion is defined as
\begin{equation}
 \cI_n[\theta]
 \equiv
 \int \dd^4x\sqrt g\,\ee^{\ii n\theta(x)},
 \qquad
 \cI_n^\dagger=\cI_{-n}.
 \label{eq:integrated-mouth}
\end{equation}
Integrating over the two mouth positions of a connected wormhole gives, at leading order in a derivative expansion,
\begin{equation}
 I_{\rm wh}
 =
 \sum_{n,m\geq1}
 \Delta_{nm}\,
 \cI_n^\dagger\cI_m,
 \label{eq:general-bilocal}
\end{equation}
where $I_{\rm wh}$ is the bilocal wormhole contribution to the Euclidean exponent, and $\Delta_{nm}$ is the kernel coupling the charge-$n$ and charge-$m$ mouth channels. The coefficient matrix collects the classical saddle contribution, the fluctuation determinants, and any correlations between different charge channels. For the constant-modulus spherical reference saddle, flux conservation makes this matrix diagonal in $n$, although we keep the more general matrix notation to allow for possible throat interactions or a broader operator basis.  In a dilute gas, disconnected wormholes exponentiate,
\begin{equation}
 Z
 =
 \int \cD g\,\cD\theta\,
 \exp\left[-S_0[g,\theta]+I_{\rm wh}[g,\theta]\right],
 \label{eq:dilute-partition}
\end{equation}
where $S_0[g,\theta]$ is the Euclidean parent-theory action in the absence of wormhole insertions, and $Z$ is the corresponding wormhole-corrected partition function. The approximation requires typical mouth separations to be large compared with the relevant throat radii and neglects interactions among wormholes.

The central state-dependence issue is already present for a diagonal kernel.  We therefore take
\begin{equation}
 I_{\rm wh}
 =
 \sum_{n\geq1}|\lambda_n|^2|\cI_n|^2,
 \label{eq:diagonal-bilocal}
\end{equation}
where Eq.~\eqref{eq:lambda-definition} has been used.  Higher-derivative mouth operators and non-diagonal charge correlations can be restored without changing the distinction between a bilocal kernel and a baby-universe state.

\subsection{Complex Hubbard--Stratonovich localization}
\label{subsec:hs-localization}

For each harmonic, the bilocal term can be localized by means of the exact complex Gaussian identity
\begin{equation}
 \ee^{|\lambda_n|^2|\cI_n|^2}
 =
 \int_{\mathbb C}\frac{\dd^2\alpha_n}{\pi}
 \exp\left[
 -|\alpha_n|^2
 +\lambda_n\alpha_n\cI_n
 +\lambda_n^*\alpha_n^*\cI_n^\dagger
 \right],
 \label{eq:complex-hs}
\end{equation}
where the auxiliary variables $\alpha_n$ are dimensionless. It is then convenient to combine $\alpha_n$ with the mouth coefficient $\lambda_n$ and define the dimension-four local quantity
\begin{equation}
 \beta_n\equiv\lambda_n\alpha_n
 =b_n\ee^{\ii\varphi_n},
 \qquad
 b_n\geq0.
 \label{eq:beta-definition}
\end{equation}
Thus, $b_n=|\beta_n|$, while its phase is
$\varphi_n=\arg\beta_n=\arg\lambda_n+\vartheta_n$ if we write
$\alpha_n=r_n\ee^{\ii\vartheta_n}$ with $r_n=|\alpha_n|$. We use
$\varphi_n$ to distinguish the phase of $\beta_n$ from the phase
$\delta_n$ of $\lambda_n$ introduced in Eq.~\eqref{eq:lambda-definition}.
In polar coordinates, the integration measure becomes
$\dd^2\alpha_n=r_n\dd r_n\dd\vartheta_n$; however, since the phase
$\vartheta_n$ is not defined at $r_n=0$, Cartesian coordinates are more
appropriate when the neighborhood of the symmetric point is considered. The wormhole-induced local Euclidean potential in a fixed $\bm\alpha$ sector, where $\bm\alpha\equiv\{\alpha_n\}_{n\geq1}$, is
\begin{align}
 V_{\rm wh}(\theta;\bm\alpha)
 &={}
 -\sum_{n\geq1}
 \left(\beta_n\ee^{\ii n\theta}
 +\beta_n^*\ee^{-\ii n\theta}\right)
 \nonumber\\
 &={}
 -2\sum_{n\geq1}b_n\cos(n\theta+\varphi_n).
 \label{eq:fixed-sector-potential}
\end{align}
The sign follows from the convention in Eq.~\eqref{eq:dilute-partition}: the localized source appears with a plus sign in the exponent and therefore with a minus sign in the Euclidean potential.  A local periodic potential is thus obtained only after the values of $\bm\alpha$ have been fixed or integrated with a specified measure.

For $N$ retained harmonics, let $\bm\cI=(\cI_1,\ldots,\cI_N)^{\mathsf T}$, $\bm\beta=(\beta_1,\ldots,\beta_N)^{\mathsf T}$, and $\Delta=(\Delta_{nm})$ be the positive-definite Hermitian bilocal kernel. Equation~\eqref{eq:complex-hs} then generalizes to
\begin{equation}
 \ee^{\bm\cI^\dagger\Delta\bm\cI}
 =
 \frac{1}{\pi^N\det\Delta}
 \int\dd^{2N}\bm\beta\,
 \exp\left[
 -\bm\beta^{\mathsf T}\Delta^{-1}\bm\beta^*
 +\bm\beta^{\mathsf T}\bm\cI
 +\bm\cI^\dagger\bm\beta^*
 \right].
 \label{eq:matrix-hs}
\end{equation}
The integration variable with covariance $\Delta$ is $\bm\beta^*$, so $\langle\bm\beta\bm\beta^\dagger\rangle=\Delta^{\mathsf T}$ in the auxiliary Gaussian. This convention agrees with the sources in Eq.~\eqref{eq:complex-hs}. For a positive-semidefinite kernel, the integral is restricted to its positive eigenspace in the $\bm\beta^*$ coordinates. A kernel with negative eigenvalues does not give a positive Gaussian localization.

\subsection{Hubbard--Stratonovich measure versus baby-universe state}
\label{subsec:state-measure}

The Gaussian factor in Eq.~\eqref{eq:complex-hs} arises solely from the Hubbard--Stratonovich localization and should not be identified with the physical baby-universe state. Each complex parameter $\alpha_n=\alpha_{n,c}+\ii\alpha_{n,s}$ packages the joint spectral coordinates of the commuting Hermitian cosine and sine combinations of the baby-universe operators \cite{AlonsoUrbano2019}. Once the normalization of these commuting Hermitian operators is fixed, we denote their simultaneous generalized eigenstates by $|\bm\alpha\rangle$.  For a density matrix $\rho_{\rm BU}$, the corresponding positive spectral measure is
\begin{equation}
 \dd\nu_{\rm BU}^{(0)}(\bm\alpha)
 =
 \dd\mu_{\rm spec}(\bm\alpha)\,
 \langle\bm\alpha|\rho_{\rm BU}|\bm\alpha\rangle,
 \qquad
 \dd\nu_{\rm BU}^{(0)}\geq0,
 \label{eq:bu-state-measure}
\end{equation}
where $\rho_{\rm BU}$ is the baby-universe density matrix, $\dd\mu_{\rm spec}$ is the joint spectral measure of the commuting baby-universe operators, and $\dd\nu_{\rm BU}^{(0)}$ is the state-induced measure before any parent-universe weighting. Neither the covariance in Eq.~\eqref{eq:matrix-hs} nor the connected wormhole saddle uniquely determines $\rho_{\rm BU}$, so a Gaussian choice for $\rho_{\rm BU}$ should be regarded as one possible baby-universe state rather than as a consequence of the Hubbard--Stratonovich transformation itself.

Parent-universe amplitudes may provide an additional weighting,
\begin{equation}
 Z[\rho_{\rm BU}]
 =
 \int\dd\nu_{\rm BU}^{(0)}(\bm\alpha)\,
 Z_{\rm parent}[\bm\alpha],
 \label{eq:parent-average}
\end{equation}
where $Z_{\rm parent}[\bm\alpha]$ is the parent-universe partition functional evaluated in a fixed $\bm\alpha$ sector, while $Z[\rho_{\rm BU}]$ denotes its average in the state $\rho_{\rm BU}$. Let $V_4$ denote the regulated Euclidean four-volume and $E_{\rm vac}(\bm\alpha)$ the vacuum-energy density in the fixed
$\bm\alpha$ sector. If a Euclidean prescription gives a real positive large-volume functional
\begin{equation}
 Z_{\rm parent}^{(E)}[\bm\alpha]
 \simeq
 \exp\left[-V_4E_{\rm vac}(\bm\alpha)\right],
 \label{eq:euclidean-parent}
\end{equation}
then a normalized effective measure can be defined by
\begin{equation}
 \dd\nu_{\rm eff}(\bm\alpha)
 =
 \frac{\ee^{-V_4E_{\rm vac}(\bm\alpha)}
 \dd\nu_{\rm BU}^{(0)}(\bm\alpha)}
 {\int\ee^{-V_4E_{\rm vac}}\dd\nu_{\rm BU}^{(0)}},
 \label{eq:effective-positive-measure}
\end{equation}
where $\dd\nu_{\rm eff}$ is the normalized measure obtained after including the parent-universe weighting.
A Lorentzian weighting $\exp[-\ii V_4E_{\rm vac}]$ is instead oscillatory and is governed by contour and stationary-phase data, which is not a probability measure.  In Ref.~\cite{Kawana2026}, the authors proposed a Lorentzian thermodynamic prescription in which the symmetric point can dominate for a class of Peccei--Quinn models, and although we do not assume this prescription here, our analysis gives an immediate spectral consequence when the resulting state is concentrated at, or has support approaching, $\bm\alpha=0$.

Equations~\eqref{eq:bu-state-measure}--\eqref{eq:effective-positive-measure} also make clear what is meant by an ensemble-level statement. An observer confined to a fixed superselection sector measures observables at a definite value of $\bm\alpha$, whereas averaging over $\bm\alpha$ refers instead to a nontrivial baby-universe state measure, an ensemble interpretation of the gravitational path integral, or simply incomplete information about the underlying sector. These possibilities are conceptually distinct and should only be compared after the relevant observable and the corresponding conditioning prescription have been specified. The fixed-sector interpretation follows the superselection description in Ref.~\cite{MarolfMaxfield2020}.

\section{Conditional mass and the quadratic spectral threshold}
\label{sec:gap}

\subsection{Mass in a fixed sector}
\label{subsec:fixed-mass}

Let $\theta_\star(\bm\alpha)$ denote a locally stable branch of minima of Eq.~\eqref{eq:fixed-sector-potential}, so that for the canonically normalized field $\phi=\mu\theta$ the corresponding conditional curvature mass is
\begin{equation}
 m_{\phi,\bm\alpha}^2
 =
 \left.
 \frac{1}{\mu^2}
 \frac{\partial^2V_{\rm wh}}{\partial\theta^2}
 \right|_{\theta=\theta_\star}
 =
 \frac{2}{\mu^2}
 \sum_{n\geq1}n^2b_n
 \cos\left(
 n\theta_\star+
 \varphi_n
 \right),
 \label{eq:general-conditional-mass}
\end{equation}
where local stability requires this quantity to be nonnegative. Throughout this section, $m_{\phi,\bm\alpha}$ denotes only the contribution generated by the wormhole harmonics, so any independent source of shift-symmetry breaking must be included separately in the sector inverse propagator; although momentum-dependent terms in the exact two-point function can shift the pole relative to this curvature mass, this distinction is irrelevant in the two-derivative quadratic approximation adopted below, where the two coincide.

For the unit harmonic, writing $\alpha=r\ee^{\ii\vartheta}$ and $\varphi=\arg\lambda_1+\vartheta$, the wormhole-induced potential becomes
\begin{equation}
 V_{\rm wh}(\theta;\alpha)
 =
 -2|\lambda_1|r\cos\left(\theta+\varphi\right),
 \label{eq:unit-potential}
\end{equation}
whose minimum lies at $\theta_\star=-\varphi$ modulo $2\pi$ and therefore gives the conditional mass
\begin{equation}
 m_{\phi,\alpha}^2
 =
 \gamma r,
 \qquad
 \gamma\equiv\frac{2|\lambda_1|}{\mu^2}.
 \label{eq:unit-mass}
\end{equation}
This relation makes the sector dependence explicit: a nonzero connected wormhole coefficient does not necessarily generate a nonzero mass in every sector. At the symmetric point $\alpha=0$, the local wormhole harmonic vanishes and hence
\begin{equation}
 m_{\phi,0}^{2}=0,
 \label{eq:symmetric-mass-zero}
\end{equation}
even though the anomaly coupling in Eq.~\eqref{eq:wz-term} remains present.

\subsection{Pushforward measure and a gap criterion}
\label{subsec:gap-criterion}

Taking $\dd\nu_{\rm eff}(\bm\alpha)$ to be a positive normalized measure, the distribution of the conditional mass squared is obtained by pushing this measure forward under the map $\bm\alpha\mapsto m_{\phi,\bm\alpha}^2$:
\begin{equation}
 P_{m^2}(x)
 =
 \int\dd\nu_{\rm eff}(\bm\alpha)\,
 \delta\left[x-m_{\phi,\bm\alpha}^2\right],
 \qquad x\geq0,
 \label{eq:mass-pushforward}
\end{equation}
where $x$ denotes the mass-squared spectral variable. If the pushforward measure is absolutely continuous, $P_{m^2}(x)$ can be interpreted as its probability density; otherwise, Eq.~\eqref{eq:mass-pushforward} is understood in the distributional sense and may also contain discrete atomic contributions. To exclude the additional covariance term generated by sector-dependent one-point functions, we subtract the one-point function in each sector before averaging the resulting connected correlators. We compare the correlators in a fixed-background approximation in which the sector-dependent vacuum energy does not source the parent metric. On a common flat background the two-derivative quadratic correlators take the form
\begin{align}
 G_{\bm\alpha,c}(p_E^2)
 &=
 \frac{Z_{\bm\alpha}}
 {p_E^2+m_{\phi,\bm\alpha}^2},
 \nonumber\\
 0<Z_{\bm\alpha}<\infty
 &\quad\text{almost everywhere},
 \qquad
 \int Z_{\bm\alpha}\dd\nu_{\rm eff}<\infty.
 \label{eq:sector-propagator}
\end{align}
Here $p_E$ is the Euclidean momentum, $Z_{\bm\alpha}$ is the pole residue in the fixed sector, and the subscript $c$ denotes a connected correlator. Their average is a Stieltjes transform,
\begin{equation}
 G_{{\rm ens},c}(p_E^2)
 =
 \int_0^\infty
 \frac{\dd\rho_{\rm ens}(x)}{p_E^2+x},
 \label{eq:ensemble-stieltjes}
\end{equation}
with positive spectral measure
\begin{equation}
 \dd\rho_{\rm ens}(x)
 =
 \int\dd\nu_{\rm eff}(\bm\alpha)\,
 Z_{\bm\alpha}\,
 \delta\left[x-m_{\phi,\bm\alpha}^2\right]\dd x.
 \label{eq:ensemble-spectral-measure}
\end{equation}

\begin{proposition}
\label{prop:gap}
Under the assumptions in Eqs.~\eqref{eq:sector-propagator} and \eqref{eq:ensemble-spectral-measure}, and provided $Z_{\bm\alpha}$ is strictly positive almost everywhere, the lower edge of the averaged spectral measure is
\begin{equation}
 m_{{\rm gap},{\rm ens}}^2
 \equiv
 \inf\operatorname{supp}\rho_{\rm ens}
 =
 \essinf_{\bm\alpha\sim\nu_{\rm eff}}
 m_{\phi,\bm\alpha}^2.
 \label{eq:gap-proposition}
\end{equation}
\end{proposition}

\begin{proof}
Let $b=\essinf m_{\phi,\bm\alpha}^2$.  For every $\epsilon>0$, the set $A_\epsilon=\{\bm\alpha:m_{\phi,\bm\alpha}^2<b+\epsilon\}$ has positive $\nu_{\rm eff}$ measure.  Since $Z_{\bm\alpha}>0$ almost everywhere, Eq.~\eqref{eq:ensemble-spectral-measure} assigns nonzero spectral weight to $[b,b+\epsilon)$.  Hence the support of $\rho_{\rm ens}$ accumulates at $b$.  Conversely, the set on which $m_{\phi,\bm\alpha}^2<b-\epsilon$ has zero measure for every $\epsilon>0$, so $\rho_{\rm ens}$ has no support below $b$.  Therefore $\inf\operatorname{supp}\rho_{\rm ens}=b$.
\end{proof}
For unit residue, $\dd\rho_{\rm ens}$ coincides with the pushforward measure represented by Eq.~\eqref{eq:mass-pushforward}; any atomic contributions are already included in its distributional definition.
An isolated pole at $x=m_*^2$ requires an atom $\dd\rho_{\rm ens}(x)\supset Z_*\delta(x-m_*^2)\dd x$. A continuous distribution instead gives a branch cut after analytic continuation.  The measure in Eq.~\eqref{eq:ensemble-spectral-measure} is the positive measure of an averaged Stieltjes transform; it need not be the K\"all\'en--Lehmann measure of a single quantum field theory with one Hilbert space.  Proposition~\ref{prop:gap} concerns the lower spectral threshold, not the mean mass.  In particular,
\begin{equation}
 \langle m_\phi^2\rangle
 =
 \int\dd\nu_{\rm eff}(\bm\alpha)\,
 m_{\phi,\bm\alpha}^2>0
 \quad\nRightarrow\quad
 m_{{\rm gap},{\rm ens}}>0.
 \label{eq:mean-not-gap}
\end{equation}

\begin{corollary}
\label{cor:support-zero}
For Eq.~\eqref{eq:unit-mass}, if the support of $\nu_{\rm eff}$ accumulates at $\alpha=0$, then
\begin{equation}
 m_{{\rm gap},{\rm ens}}=0.
 \label{eq:gapless-corollary}
\end{equation}
This holds even if $\nu_{\rm eff}(\{0\})=0$ and $m_{\phi,\alpha}>0$ almost surely.
\end{corollary}

The gap functional is not continuous under weak convergence of measures. A simple example illustrates this subtlety: a sequence of full-support distributions can converge weakly to $\delta^{(2)}(\alpha-\alpha_\star)$ with $\alpha_\star\neq0$, even though every distribution in the sequence has a vanishing essential-infimum gap while the limiting delta measure has a nonzero one. This shows that the order in which the infinite-volume and spectral-support limits are taken must be specified explicitly.

For interacting sectors, the same argument applies after replacing $m_{\phi,\bm\alpha}^2$ by the lower edge of each sector's positive spectral measure, provided every neighborhood of that edge carries nonzero spectral weight. In perturbative flat-space gravity with a nonzero coupling in Eq.~\eqref{eq:wz-term}, the axion couples to two on-shell gravitons \cite{AlexanderEtAl2025,Ema:2021fdz}. The corresponding self-energy cut starts at zero invariant mass because the gravitons are massless. The axion response can therefore contain continuum weight down to zero even in a fixed sector with positive curvature mass, while its massive excitation can appear as a narrow resonance at weak coupling. The spectral threshold and nonanalytic terms computed below refer to Eq.~\eqref{eq:sector-propagator}.  If the gravitational weighting is complex, Eqs.~\eqref{eq:mass-pushforward}--\eqref{eq:gap-proposition} are not probability or positivity statements and a separate contour analysis is required.

\subsection{Universal behavior at a zero threshold}
\label{subsec:threshold-behavior}

The nonanalyticity of the averaged propagator is fixed by the amount of spectral weight near the lower spectral edge.  For unit residue, assume that the pushforward measure is absolutely continuous in a neighborhood of $x=0$, with density
\begin{equation}
 P_{m^2}(x)=A x^{\eta}+o(x^{\eta}),
 \qquad x\to0^+,
 \qquad A>0,
 \qquad \eta>0,
 \label{eq:threshold-density-general}
\end{equation}
where $A$ is the leading threshold coefficient and $\eta$ is the threshold exponent, and we also write $s\equiv p_E^2$ for the Euclidean momentum squared. For $0<\eta<1$, subtracting the finite zero-momentum value from Eq.~\eqref{eq:ensemble-stieltjes} gives
\begin{align}
 G_{{\rm ens},c}(s)-G_{{\rm ens},c}(0)
 &=-s\int_0^\infty
 \frac{P_{m^2}(x)}{x(x+s)}\dd x
 \nonumber\\
 &=-\frac{\pi A}{\sin(\pi\eta)}s^{\eta}
 +o(s^{\eta}),
 \qquad s\to0^+.
 \label{eq:fractional-threshold}
\end{align}
At the marginal value $\eta=1$, assuming in addition $P_{m^2}(x)=Ax+\cO(x^{1+\delta})$ for some $\delta>0$, the corresponding result is
\begin{equation}
 G_{{\rm ens},c}(s)
 =G_{{\rm ens},c}(0)
 +A s\ln\left(\frac{s}{x_0}\right)
 +\cO(s),
 \label{eq:linear-threshold-log}
\end{equation}
where $x_0>0$ is a reference mass-squared scale and the analytic $\cO(s)$ term depends on the spectrum away from threshold.  Equations~\eqref{eq:fractional-threshold} and \eqref{eq:linear-threshold-log} follow by setting $x=st$ in the threshold part of the integral.  They show that a vanishing density at the origin does not restore an isolated pole. A radial mass map on a $d$-dimensional real coefficient space gives one example, where $d$ counts the independent real components. If the state density is nonzero at the origin and $m_{\phi,\bm\alpha}^2\propto r$
with $r=|\bm\alpha|$, then the radial Jacobian gives
\begin{equation}
 P_{m^2}(x)\propto x^{d-1}.
 \label{eq:dimension-threshold}
\end{equation}
Note that for a single complex harmonic, $d=2$, and therefore $\eta=1$, leading to the $s\ln s$ behavior displayed explicitly below.

\subsection{Gaussian baby-universe state}
\label{subsec:gaussian-example}

Consider the unit harmonic and the normalized rotationally invariant state
\begin{equation}
 \dd\nu_\sigma(\alpha)
 =
 \frac{\dd^2\alpha}{\pi\sigma^2}
 \exp\left(-\frac{|\alpha|^2}{\sigma^2}\right),
 \label{eq:gaussian-alpha}
\end{equation}
where $\sigma>0$ controls the width of the Gaussian measure in the complex $\alpha$ plane. The radial density is
\begin{equation}
 p_r(r)=\frac{2r}{\sigma^2}\ee^{-r^2/\sigma^2},
 \qquad r\geq0.
 \label{eq:rayleigh-radial}
\end{equation}
Define the mass-squared scale
\begin{equation}
 m_\sigma^2\equiv\gamma\sigma
 =\frac{2|\lambda_1|\sigma}{\mu^2}, 
 \label{eq:msigma-definition}
\end{equation}
and by using Eq.~\eqref{eq:unit-mass}, the pushforward density is
\begin{equation}
 P_{m^2}(x)
 =
 \frac{2x}{m_\sigma^4}
 \exp\left(-\frac{x^2}{m_\sigma^4}\right),
 \qquad x\geq0.
 \label{eq:gaussian-mass-density}
\end{equation}
It is normalized and has moments
\begin{equation}
 \langle(m_\phi^2)^k\rangle
 =
 (m_\sigma^2)^k
 \Gamma\left(1+\frac{k}{2}\right),
 \qquad k>-2,
 \label{eq:gaussian-moments}
\end{equation}
and in particular, one can write
\begin{equation}
 \langle m_\phi^2\rangle
 =\frac{\sqrt\pi}{2}m_\sigma^2,
 \qquad
 \operatorname{Var}(m_\phi^2)
 =m_\sigma^4\left(1-\frac{\pi}{4}\right).
 \label{eq:gaussian-mean-variance}
\end{equation}
The corresponding density for the mass itself is
\begin{equation}
 P_m(m)
 =
 \frac{4m^3}{m_\sigma^4}
 \exp\left(-\frac{m^4}{m_\sigma^4}\right),
 \qquad m\geq0.
 \label{eq:gaussian-mass-density-m}
\end{equation}
The mass distribution obtained from Eq.~\eqref{eq:gaussian-mass-density} is shown in Fig.~\ref{fig:mass-distribution}. Although almost every fixed sector is massive and the mean mass squared in Eq.~\eqref{eq:gaussian-mean-variance} is positive, the support of the distribution still extends down to zero. It then follows from Corollary~\ref{cor:support-zero} that $m_{{\rm gap},{\rm ens}}=0$. In addition, Eq.~\eqref{eq:gaussian-mass-density} gives $A=2/m_\sigma^4$ and $\eta=1$, which already determines the logarithmic behavior in Eq.~\eqref{eq:linear-threshold-log} before the full integral is evaluated.

\begin{figure}[t]
 \includegraphics[width=\columnwidth]{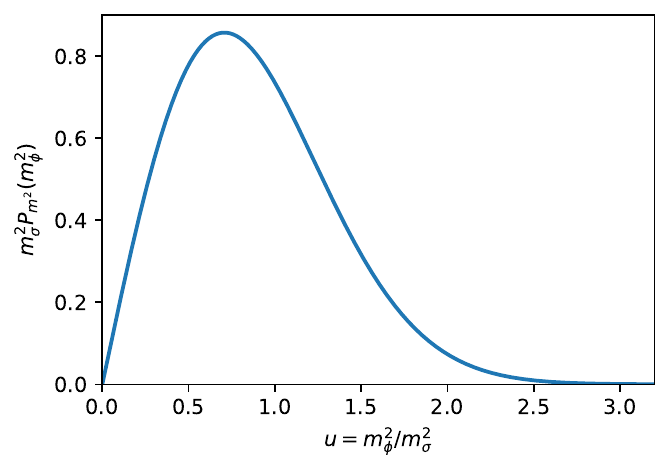}
 \caption{ Dimensionless pushforward density for the Gaussian state in Eq.~\eqref{eq:gaussian-alpha}.  With $u=m_\phi^2/m_\sigma^2$, the plotted density is
 $m_\sigma^2P_{m^2}(m_\sigma^2u)=2u\exp(-u^2)$. The density vanishes at the origin, but every neighborhood of the origin has nonzero measure.
 }
 \label{fig:mass-distribution}
\end{figure}

For unit residue, the averaged Euclidean propagator can also be evaluated exactly:
\begin{align}
 G_{{\rm ens},c}(s)
 &\equiv
 \int_0^\infty\dd x\,
 \frac{P_{m^2}(x)}{s+x}
 \nonumber\\
 &=
 \frac{1}{m_\sigma^2}\,
 \cG\left(\frac{s}{m_\sigma^2}\right),
 \qquad s=p_E^2.
 \label{eq:gaussian-propagator-definition}
\end{align}
and by writing $z\equiv s/m_\sigma^2>0$, the dimensionless function $\cG$ is
\begin{equation}
 \cG(z)
 =
 \sqrt\pi
 -z\ee^{-z^2}
 \left[
 \pi\operatorname{erfi}(z)-\operatorname{Ei}(z^2)
 \right],
 \label{eq:gaussian-propagator-exact}
\end{equation}
where $\operatorname{erfi}$ and $\operatorname{Ei}$ denote the imaginary error function and the exponential integral, respectively. For $s\to0^+$, one can obtain
\begin{equation}
 G_{{\rm ens},c}(s)
 =
 \frac{\sqrt\pi}{m_\sigma^2}
 +\frac{2s}{m_\sigma^4}
 \left[
 \ln\left(\frac{s}{m_\sigma^2}\right)
 +\frac{\gamma_E}{2}
 \right]
 +\cO\left(\frac{s^2}{m_\sigma^6}\right),
 \label{eq:gaussian-propagator-small-s}
\end{equation}
where $\gamma_E$ is the Euler--Mascheroni constant. The zero-momentum value is finite because the spectral density vanishes linearly at threshold, but the $s\ln s$ term is nonanalytic. After analytic continuation to Lorentzian momentum, Eq.~\eqref{eq:gaussian-propagator-exact} develops a branch cut starting at $p^2=0$ rather than an isolated massive pole, as illustrated by the momentum dependence shown in Fig.~\ref{fig:ensemble-propagator}; the integral leading to Eqs.~\eqref{eq:gaussian-propagator-exact} and \eqref{eq:gaussian-propagator-small-s} is given in \ref{app:gaussian}.

\begin{figure}[t]
 \includegraphics[width=\columnwidth]{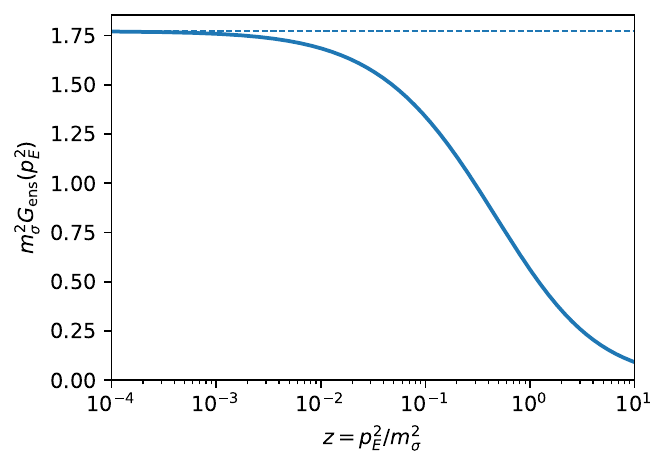}
 \caption{Exact dimensionless averaged propagator $m_\sigma^2G_{{\rm ens},c}$ as a function of $z=p_E^2/m_\sigma^2$.  The dashed line is the finite zero-momentum limit $\sqrt\pi$.  The nonanalytic approach to this limit is given by Eq.~\eqref{eq:gaussian-propagator-small-s}.}
 \label{fig:ensemble-propagator}
\end{figure}

\section{Prepared states and the weight near the spectral threshold}
\label{sec:prepared-states}

\subsection{A state construction for the Gaussian measure}
\label{subsec:state-construction}

A Gaussian measure can be motivated by specifying a baby-universe state rather than by identifying the Hubbard--Stratonovich weight with a density matrix. We use the oscillator realization of the dilute baby-universe algebra, restricted to one charge-conjugate pair \cite{AlonsoUrbano2019}. Let $a_+$ and $a_-$ obey $[a_\epsilon,a_{\epsilon'}^\dagger]=\delta_{\epsilon\epsilon'}$ and annihilate the Fock vacuum $|0,0\rangle$. The operator coupled to a charged mouth and its Hermitian components are
\begin{equation}
 \widehat C=a_++a_-^\dagger,
 \qquad
 \widehat X=\frac{\widehat C+\widehat C^\dagger}{2},
 \qquad
 \widehat Y=\frac{\widehat C-\widehat C^\dagger}{2\ii}.
 \label{eq:bu-oscillator-pair}
\end{equation}
The cancellation between the two oscillator commutators gives $[\widehat C,\widehat C^\dagger]=0$, and hence $[\widehat X,\widehat Y]=0$. Their joint spectral coordinate is $\alpha=\alpha_c+\ii\alpha_s$. This construction concerns the commuting mouth algebra; it does not assume that $\widehat X$ and $\widehat Y$ are a canonical position--momentum pair.

Consider the normalized displaced two-mode squeezed state
\begin{align}
 |\psi\rangle&=D_+(z_+)D_-(z_-)S_2(\zeta)|0,0\rangle,
 \nonumber\\
 D_\epsilon(z)&=\exp(z a_\epsilon^\dagger-z^*a_\epsilon),
 \qquad
 S_2(\zeta)=\exp\!\left[\zeta(a_+a_--a_+^\dagger a_-^\dagger)\right],
 \quad \zeta\in\mathbb R.
 \label{eq:bu-state-family}
\end{align}
Its mean and width are denoted by
\begin{equation}
 \alpha_0=z_++z_-^*,
 \qquad \sigma=\ee^{-\zeta}>0.
 \label{eq:bu-state-parameters}
\end{equation}
The joint characteristic function, derived in \ref{app:state-spectrum}, is
\begin{equation}
 \langle\psi|\ee^{\ii(k_c\widehat X+k_s\widehat Y)}|\psi\rangle
 =\exp\!\left[
 \ii(k_c\operatorname{Re}\alpha_0+k_s\operatorname{Im}\alpha_0)
 -\frac{\sigma^2}{4}(k_c^2+k_s^2)\right].
 \label{eq:bu-characteristic}
\end{equation}
Consequently the spectral probability measure of this state is
\begin{equation}
 \dd\nu_{\alpha_0,\sigma}(\alpha)
 =\frac{\dd^2\alpha}{\pi\sigma^2}
 \exp\!\left[-\frac{|\alpha-\alpha_0|^2}{\sigma^2}\right].
 \label{eq:displaced-gaussian-state}
\end{equation}
The Fock vacuum corresponds to $\alpha_0=0$ and $\sigma=1$ in the normalization of Eq.~\eqref{eq:bu-oscillator-pair}. Product coherent oscillator states have $\sigma=1$ and generally nonzero $\alpha_0$. Changing the width in this pure-state family requires squeezing; it is not an independent freedom of the Fock vacuum. In particular, a coherent oscillator state is not a fixed $\alpha$ sector. Exact simultaneous eigenstates of $\widehat X$ and $\widehat Y$ are generalized, nonnormalizable spectral states.

This is a kinematical realization of a chosen state in the oscillator description, consistent with the interpretation reviewed in Ref.~\cite{HebeckerReview2018}. It neither derives a baby-universe Hamiltonian nor selects a unique state of quantum gravity. Equation~\eqref{eq:displaced-gaussian-state} supplies an explicit instance of $\dd\nu_{\rm BU}^{(0)}$ in Eq.~\eqref{eq:bu-state-measure}. For the normalized sector observables considered below we use this state measure without an additional parent-universe weighting. If Eq.~\eqref{eq:effective-positive-measure} is instead adopted, its reweighting must be included explicitly; the resulting effective measure need not remain Gaussian. Lorentzian evolution within a superselection sector does not by itself change its $\alpha$ eigenvalue.

\subsection{Displaced states and spectral weight at finite scales}
\label{subsec:displaced-spectrum}

The support of the spectral measure specifies a threshold but does not quantify the probability of obtaining a mass below a finite scale. The state in Eq.~\eqref{eq:displaced-gaussian-state} allows these two questions to be separated analytically. After a rotation of the complex $\alpha$ plane, the radial distribution depends only on
\begin{equation}
 q=\frac{|\alpha_0|}{\sigma},
 \qquad
 u=\frac{x}{m_\sigma^2},
 \qquad x=m_\phi^2,
 \label{eq:displacement-variables}
\end{equation}
where $m_\sigma$ is defined in Eq.~\eqref{eq:msigma-definition}. Angular integration gives the noncentral radial density
\begin{equation}
 P_{m^2}^{(q)}(x)
 =\frac{2x}{m_\sigma^4}
 \exp\!\left[-\frac{x^2}{m_\sigma^4}-q^2\right]
 I_0\!\left(\frac{2qx}{m_\sigma^2}\right),
 \qquad x\geq0.
 \label{eq:noncentral-mass-density}
\end{equation}
Here $I_j$ is the modified Bessel function of the first kind. The central result in Eq.~\eqref{eq:gaussian-mass-density} is recovered for $q=0$. The noncentral and central radial laws are the Rice and Rayleigh laws, respectively \cite{Rice1945}. For every finite $q$, the density in Eq.~\eqref{eq:noncentral-mass-density} has support arbitrarily close to zero, and hence the ensemble spectral gap is zero. Its threshold expansion is nevertheless state dependent:
\begin{equation}
 P_{m^2}^{(q)}(x)
 =\frac{2\ee^{-q^2}}{m_\sigma^4}x
 \left[1+(q^2-1)\frac{x^2}{m_\sigma^4}
 +\cO\!\left(\frac{x^4}{m_\sigma^8}\right)\right].
 \label{eq:noncentral-threshold}
\end{equation}
Thus the coefficient of the logarithmic nonanalyticity in Eq.~\eqref{eq:linear-threshold-log} is suppressed by $\ee^{-q^2}$ at fixed $m_\sigma$. This suppression changes the weight near the threshold without moving the threshold.

The probability of a mass below a specified scale $M>0$ is
\begin{equation}
 F_q(M)
 \equiv\Pr(m_\phi<M)
 =\int_0^{M^2/m_\sigma^2}\!\dd u\,
 2u\ee^{-u^2-q^2}I_0(2qu).
 \label{eq:finite-scale-cdf}
\end{equation}
For the central state this integral is elementary,
\begin{equation}
 F_0(M)=1-\exp[-(M/m_\sigma)^4],
 \label{eq:central-finite-cdf}
\end{equation}
and for any fixed finite displacement it obeys
\begin{equation}
 F_q(M)=\ee^{-q^2}(M/m_\sigma)^4
 +\cO\!\left((M/m_\sigma)^8\right),
 \qquad M/m_\sigma\to0.
 \label{eq:finite-scale-small-mass}
\end{equation}
For example, Eq.~\eqref{eq:central-finite-cdf} assigns probability approximately $10^{-12}$ to $m_\phi<10^{-3}m_\sigma$, despite its vanishing spectral gap. This probability is a weight over sectors. It is not the fraction of the dark-matter energy carried by such masses within one sector.

To compare different states without merely changing their mean mass squared, define
\begin{align}
 h(q)&=\frac{\sqrt\pi}{2}\ee^{-q^2/2}
 \left[(1+q^2)I_0(q^2/2)+q^2 I_1(q^2/2)\right],
 \nonumber\\
 \overline x&\equiv\langle m_\phi^2\rangle=m_\sigma^2h(q).
 \label{eq:noncentral-mean}
\end{align}
Holding $\overline x$ fixed therefore requires $m_\sigma^2=\overline x/h(q)$. The small-mass weight in this comparison becomes
\begin{equation}
 F_q(M)
 =\ee^{-q^2}h(q)^2\frac{M^4}{\overline x^2}
 +\cO\!\left(\frac{M^8}{\overline x^4}\right).
 \label{eq:matched-mean-threshold}
\end{equation}
The states retain the same mean mass squared and the same zero spectral threshold, but their low-mass probabilities differ. Their dimensionless shape parameter $q$ is independent of the overall mouth normalization, whereas $m_\sigma^2$ depends on the product $|\lambda_1|\sigma$. A mass-distribution calculation cannot separate these two factors without further input.

A fixed sector chosen so that $m_\phi^2=\overline x$ provides a conditional, gapped comparison with the same mean. For finite $q$, the distribution in Eq.~\eqref{eq:noncentral-mass-density} has no atomic contribution and therefore no isolated pole in its quadratic averaged propagator. At fixed $\overline x$, its $q\to\infty$ limit instead converges weakly to the fixed-mass measure. This gives an explicit realization of the noncommuting spectral-support and localization limits discussed after Corollary~\ref{cor:support-zero}.

\begin{figure}[t]
 \centering
 \includegraphics[width=\textwidth]{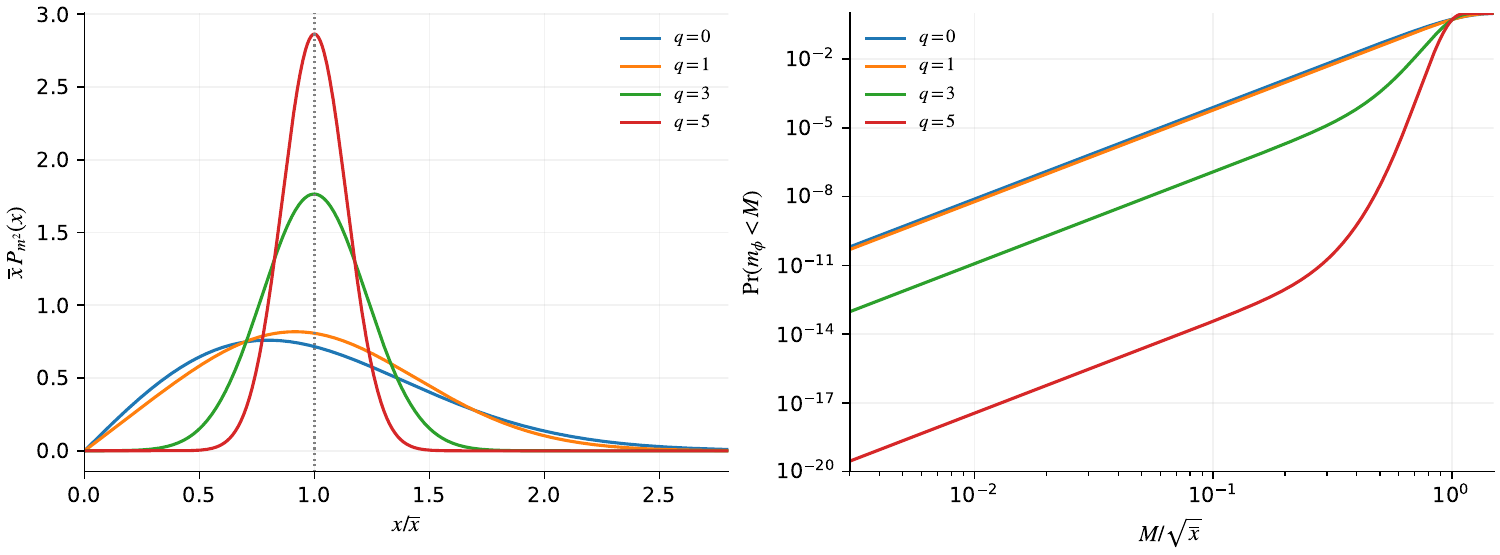}
 \caption{State dependence at matched mean mass squared $\overline x$.
 Left: the density in Eq.~\eqref{eq:noncentral-mass-density}, with its scale adjusted using Eq.~\eqref{eq:noncentral-mean}; the dotted line marks the conditional fixed-sector value $x=\overline x$.
 Right: the finite-scale probability in Eq.~\eqref{eq:finite-scale-cdf}.
 All four continuous states have zero ensemble spectral gap, while their low-mass weights differ. The fixed-sector comparison has spectral measure $\delta(x-\overline x)\dd x$ and a nonzero conditional gap.}
 \label{fig:state-spectral-comparison}
\end{figure}

\subsection{Control of a second harmonic}
\label{subsec:second-harmonic-bound}

The unit-harmonic approximation need not remain valid in the small-mass tail even when the charge-two mouth coefficient is suppressed. To quantify this limitation, retain two terms in Eq.~\eqref{eq:fixed-sector-potential} and shift the angular coordinate to write
\begin{equation}
 V(\chi)=-A\cos\chi-B\cos(2\chi+\delta),
 \quad A=2b_1>0,\quad B=2b_2\geq0,
 \quad \delta=\varphi_2-2\varphi_1.
 \label{eq:two-harmonic-audit-potential}
\end{equation}
Here $\chi=\theta+\varphi_1$, and $A$ and $B$ have mass dimension four. If $\chi_\star$ is any global minimum on the circle, its curvature obeys
\begin{equation}
 \max\{0,A-4B\}\leq V''(\chi_\star)\leq A+4B,
 \qquad
 \left|m_{12}^2-\frac{A}{\mu^2}\right|
 \leq\frac{4B}{\mu^2},
 \label{eq:two-harmonic-uniform-bound}
\end{equation}
where $m_{12}^2=V''(\chi_\star)/\mu^2$. The bound holds at the displaced global minimum. Its proof is given in \ref{app:two-harmonic-bound}. It follows that
\begin{equation}
 \frac{4b_2}{b_1}\leq\epsilon
 \quad\Longrightarrow\quad
 \frac{|m_{12}^2-m_1^2|}{m_1^2}\leq\epsilon,
 \qquad m_1^2=\frac{2b_1}{\mu^2},
 \label{eq:harmonic-tolerance-certificate}
\end{equation}
for any tolerance $\epsilon>0$. This criterion is uniform in the relative phase. It does not, however, control sectors in which $b_1$ approaches zero at fixed nonzero $b_2$: at $A=0$, direct minimization gives $m_{12}^2=4B/\mu^2$.

The weight of sectors for which Eq.~\eqref{eq:harmonic-tolerance-certificate} does not certify the chosen accuracy can also be calculated. Take independent Gaussian states for the two channels, with widths $\sigma_1$ and $\sigma_2$, a displacement $q=|\alpha_{0,1}|/\sigma_1$ in the first channel, and zero displacement in the second. Define the dimensionless harmonic hierarchy
\begin{equation}
 \kappa_h=\frac{|\lambda_2|\sigma_2}{|\lambda_1|\sigma_1}.
 \label{eq:harmonic-hierarchy}
\end{equation}
For $\kappa_h>0$, the uncertified probability is exactly
\begin{equation}
 \Pr\!\left(\frac{4b_2}{b_1}>\epsilon\right)
 =\frac{16\kappa_h^2}{\epsilon^2+16\kappa_h^2}
 \exp\!\left[-\frac{q^2\epsilon^2}{\epsilon^2+16\kappa_h^2}\right].
 \label{eq:second-harmonic-uncertified-weight}
\end{equation}
Equation~\eqref{eq:second-harmonic-uncertified-weight} is the probability of an uncertified sector. For $q=0$, $\kappa_h=10^{-3}$ and $\epsilon=0.1$, it gives $1.59744\times10^{-3}$. The actual probability of exceeding the same relative mass-squared tolerance is $7.97654\times10^{-4}$, from Eq.~\eqref{eq:two-harmonic-actual-error}.

Selecting sectors by a small unit-harmonic mass changes the diagnostic. For the central state, let $Y=M^2/m_\sigma^2$ and $t=\epsilon^2/(16\kappa_h^2)$. Truncating the integral in Eq.~\eqref{eq:noncentral-radius-laplace} at $R_1=Y$ gives
\begin{equation}
 \Pr\!\left(\frac{4b_2}{b_1}>\epsilon\,\middle|\,m_1<M\right)
 =
 \frac{1-\ee^{-(1+t)Y^2}}{(1+t)(1-\ee^{-Y^2})}.
 \label{eq:two-harmonic-conditional-certificate}
\end{equation}
The conditional probability tends to one as $M\to0$ at fixed positive $\kappa_h$. For the benchmark above and $M=10^{-3}m_\sigma$ it is approximately $0.99999999969$. The actual error also exceeds any fixed tolerance with conditional probability approaching one in this limit. Indeed, at least one of the two pure second-harmonic minima has $\cos\chi\geq0$, so $V(\chi_\star)\leq-B$. Combining this inequality with Eq.~\eqref{eq:two-harmonic-stationarity-audit} yields $V''(\chi_\star)\geq4B-3A$. For $m_1<M$, the condition $B>(1+\epsilon/4)\mu^2M^2$ therefore guarantees that the relative error exceeds $\epsilon$. Independence and the continuous second-channel law give the conditional limit.

For a finite number of harmonics, the curvature at any stable minimum satisfies
\begin{equation}
 0\leq m_{\phi,\bm\alpha}^2
 \leq\frac{2}{\mu^2}\sum_{n=1}^{N}n^2b_n.
 \label{eq:multi-harmonic-origin-bound}
\end{equation}
A positive state with support in every neighborhood of the joint coefficient origin consequently has a zero ensemble gap. Support approaching $\alpha_1=0$ alone does not establish the same result after higher harmonics are retained. Nor does Eq.~\eqref{eq:multi-harmonic-origin-bound} determine a universal threshold exponent. As an explicit example, $A=4B$ and $\delta=\pi$ give
\begin{equation}
 V(\chi)=2B(\cos\chi-1)^2-3B,
 \label{eq:quartic-critical-harmonic-example}
\end{equation}
with a quartic global minimum and zero quadratic curvature away from the joint coefficient origin. The radial exponent in Eq.~\eqref{eq:dimension-threshold} applies only under its stated radial mass-map assumption; it cannot be inferred for the general multi-harmonic potential by counting coefficients.

\subsection{Global-minimum mass distribution}
\label{subsec:two-harmonic-distribution}

The critical configuration in Eq.~\eqref{eq:quartic-critical-harmonic-example} lies at nonzero amplitudes. The distribution of curvature near a quartic absolute minimum is known to yield a frequency density proportional to the fourth power of the frequency in random-potential models \cite{GurarieChalker2003}. For the wormhole potential, the exact global-minimum condition on the circle allows us to compute the full curvature distribution, its state-dependent threshold coefficient, and the crossover from the unit harmonic.

Use the independent Gaussian channels specified above, and let $a_\sigma=\mu^2m_\sigma^2$ with $\sigma=\sigma_1$ in Eq.~\eqref{eq:msigma-definition}. Define $c_n=2\beta_n/a_\sigma$ for $n=1,2$. The first coefficient has complex Gaussian width one and mean of modulus $q$, and the second has width $\kappa_h$ and zero mean. For a global minimum at $\theta_\star$, write $c_1\ee^{\ii\theta_\star}=y-4U-2\ii W$ and $c_2\ee^{2\ii\theta_\star}=U+\ii W$, with $y=m_{12}^2/m_\sigma^2$. These relations implement stationarity. The dimensionless potential difference at an angular displacement $\xi$ is
\begin{equation}
 \frac{V_{\rm wh}(\theta_\star+\xi)-V_{\rm wh}(\theta_\star)}{a_\sigma}
 =(1-\cos\xi)\,[y-2U+2U\cos\xi-2W\sin\xi].
 \label{eq:two-harmonic-global-difference}
\end{equation}
Minimizing the bracket over the circle gives the exact condition $y-2U\geq 2\sqrt{U^2+W^2}$. For $y>0$ this condition defines the domain $\mathcal D_y=\{(U,W):U\leq y/4-W^2/y\}$. Its interior gives a unique global minimum modulo $2\pi$. The boundary and the configurations with $y=0$ have zero Gaussian probability.

The Jacobian of the map from $(\theta_\star,y,U,W)$ to the four real coefficients is $y$, as derived in \ref{app:two-harmonic-bound}. Integrating the Gaussian weight over $\theta_\star$ gives the density of $y$,
\begin{equation}
 f_{12}(y)
 =\frac{2y\ee^{-q^2}}{\pi\kappa_h^2}
 \int_{\mathcal D_y}\!\dd U\,\dd W\,
 \exp\!\left[-L^2-\frac{U^2+W^2}{\kappa_h^2}\right]
 I_0(2qL),
 \qquad
 L=\sqrt{(y-4U)^2+4W^2}.
 \label{eq:two-harmonic-exact-density}
\end{equation}
Here $L=A/a_\sigma$. The density is normalized because almost every coefficient pair enters the change of variables once. The mass CDF is $F_{12}(M)\equiv\Pr(m_{12}<M)=\int_0^{M^2/m_\sigma^2}f_{12}(y)\dd y$.

At fixed $\kappa_h>0$ and finite $q$, rescale $W=\sqrt{y}\,w$ in Eq.~\eqref{eq:two-harmonic-exact-density}. The limiting domain is $U\leq-w^2$. Thus the leading configurations lie near $W=0$, $U<0$, with $L$ approaching $4|U|$, as in Eq.~\eqref{eq:quartic-critical-harmonic-example}. The curvature Jacobian contributes a factor $y$, and the width in $W$ contributes $\sqrt{y}$. Evaluating the limiting integral as described in \ref{app:two-harmonic-bound} yields
\begin{align}
 F_{12}(M)&=C(q,\kappa_h)\left(\frac{M}{m_\sigma}\right)^5
 +o\!\left((M/m_\sigma)^5\right),
 \label{eq:two-harmonic-mass-tail}\\
 C(q,\kappa_h)&=
 \frac{8\ee^{-q^2}}{5\pi\kappa_h^2}
 \int_0^\infty\!\dd v\,\sqrt{v}\,
 \exp[-(16+\kappa_h^{-2})v^2] I_0(8qv),
 \label{eq:two-harmonic-tail-coefficient}\\
 C(0,\kappa_h)&=
 \frac{4\Gamma(3/4)}{5\pi\sqrt{\kappa_h}\,(1+16\kappa_h^2)^{3/4}}.
 \label{eq:two-harmonic-central-coefficient}
\end{align}
Expanding $I_0$ in Eq.~\eqref{eq:two-harmonic-tail-coefficient} also gives
\begin{equation}
 C(q,\kappa_h)=C(0,\kappa_h)\ee^{-q^2}\,
 {}_1F_1\!\left(\frac34;1;\frac{16q^2\kappa_h^2}{1+16\kappa_h^2}\right).
 \label{eq:two-harmonic-displaced-coefficient}
\end{equation}
The coefficient is positive for every finite $q$ and $\kappa_h>0$. With unit residues, the associated mass-squared density satisfies Eq.~\eqref{eq:threshold-density-general} with $\eta=3/2$ and coefficient $D_{12}=5C(q,\kappa_h)/(2m_\sigma^5)$. In Eq.~\eqref{eq:ensemble-stieltjes}, subtracting the constant and linear terms then gives
\begin{equation}
 G_{{\rm ens},c}(s)=G_{{\rm ens},c}(0)
 -s\langle x^{-2}\rangle_{12}+\pi D_{12}s^{3/2}+o(s^{3/2}),
 \qquad s\to0^+,
 \label{eq:two-harmonic-propagator-threshold}
\end{equation}
where $\langle\cdot\rangle_{12}$ denotes the two-harmonic pushforward average at $x=m_{12}^2$. The linear coefficient is finite, whereas the unit-harmonic behavior in Eq.~\eqref{eq:linear-threshold-log} has a logarithmically divergent slope.

\subsection{Weak-hierarchy crossover}
\label{subsec:two-harmonic-crossover}

For $\kappa_h\ll1$, the crossover occurs when $M$ is of order $m_\sigma\sqrt{\kappa_h}$. To resolve it, set $z=M^2/(\kappa_h m_\sigma^2)$ and take $\kappa_h\to0$ at fixed $z$ and $q$. Scaling $(y,U,W)$ by $\kappa_h$ in Eq.~\eqref{eq:two-harmonic-exact-density} gives
\begin{equation}
 \frac{F_{12}(M)}{\ee^{-q^2}\kappa_h^2}
 \longrightarrow
 \mathcal F(z)=\mathbb E\!\left[
 \bigl(z^2-4R^2(1+\cos\psi)^2\bigr)_+
 \right],
 \label{eq:two-harmonic-crossover}
\end{equation}
where $(a)_+=\max(a,0)$, $R^2$ is unit exponential, and $\psi$ is uniform on $[0,2\pi)$ and independent of $R$. The scaled first-channel Gaussian--Bessel factor tends to one and is bounded by $\ee^{q^2}$, so the second-channel Gaussian justifies the limit under the integral at fixed $z$. In the scaled polar coordinates the global-minimum condition is $y/\kappa_h\geq2R(1+\cos\psi)$, and integration over $y/\kappa_h$ gives Eq.~\eqref{eq:two-harmonic-crossover}. At small $z$, rescaling the imaginary part as in the threshold calculation gives $\mathcal F(z)\sim 4\Gamma(3/4)z^{5/2}/(5\pi)$. At large $z$ the expectation gives $\mathcal F(z)=z^2-6+\cO(\ee^{-z^2/16})$. The constant $6$ is the expectation of $4R^2(1+\cos\psi)^2$. The two limits are shown in Fig.~\ref{fig:two-harmonic-crossover}.

For $\sqrt{\kappa_h}\ll M/m_\sigma\ll1$ at fixed $q$, Eq.~\eqref{eq:two-harmonic-uniform-bound} sandwiches $F_{12}$ between the single-harmonic CDFs at curvature cutoffs $(\mu^2M^2-4B)_+$ and $\mu^2M^2+4B$, averaged over $B$. Since $B/(\kappa_h a_\sigma)$ has a fixed Rayleigh law and the density of $A/a_\sigma$ evaluated at $v$ is bounded by $2v$, this gives
\begin{equation}
 F_{12}(M)=\ee^{-q^2}Y^2
 \left[1+\cO_q(Y^2+\kappa_h/Y)\right],
 \qquad Y=M^2/m_\sigma^2,\quad \kappa_h\ll Y\ll1.
 \label{eq:two-harmonic-intermediate}
\end{equation}
The $M^4$ behavior is therefore recovered in an intermediate range at weak hierarchy. The $M^5$ law in Eq.~\eqref{eq:two-harmonic-mass-tail} applies at sufficiently small $M$ for any fixed positive $\kappa_h$. The threshold remains zero in the quadratic ensemble, while its low-mass weight and its propagator nonanalyticity depend on the order of the low-mass and weak-hierarchy limits.

\begin{figure}[!htbp]
 \centering
 \includegraphics[width=0.69\columnwidth]{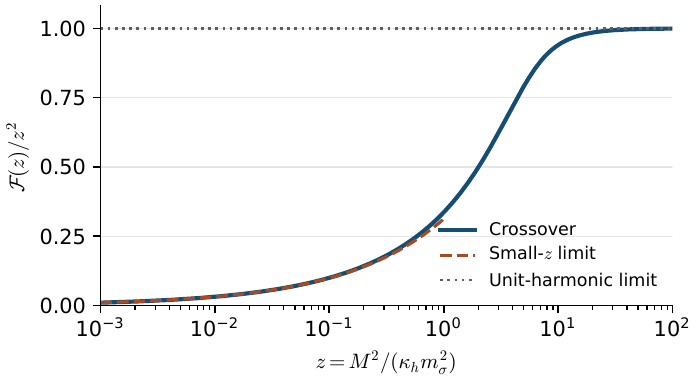}
 \caption{Weak-hierarchy crossover in Eq.~\eqref{eq:two-harmonic-crossover}, divided by the unit-harmonic scaling $z^2$. The small-$z$ curve approaches the two-harmonic threshold law, and the large-$z$ ratio approaches unity.}
 \label{fig:two-harmonic-crossover}
\end{figure}

\section{Discussion and conclusions}
\label{sec:conclusions}
The spherical wormhole calculation fixes a bilocal kernel, subject to the determinant and contour qualifications in Sec.~\ref{subsec:prefactor}. It does not select a baby-universe state. On the $O(4)$ background the Pontryagin density vanishes, so the local anomaly interaction does not modify the classical exponent. Statements about fermion selection rules or fluctuation phases remain conditional on the corresponding determinant problem; no numerical one-loop correction is claimed here.

The distinction between a conditional mass and an ensemble threshold is expressed by Proposition~\ref{prop:gap}. The explicit state construction in Sec.~\ref{sec:prepared-states} supplies a physical realization of the Gaussian examples within the assumed oscillator Hilbert space. It also identifies the assumptions: a normalized state is chosen, its diagonal measure is used, and the parent observables are normalized separately in each sector. A Euclidean weighting by the parent partition function would define a different measure. The present calculation does not resolve the general gravitational measure problem.

For the independent Gaussian channels considered here, the exact global-minimum condition gives the two-harmonic law in Eq.~\eqref{eq:two-harmonic-mass-tail}. The leading sectors lie near quartic minima at nonzero harmonic amplitudes. In the weak hierarchy, Eq.~\eqref{eq:two-harmonic-crossover} connects this tail to the unit-harmonic regime. The threshold of the quadratic average remains zero, while the new low-mass weight gives the nonanalytic term in Eq.~\eqref{eq:two-harmonic-propagator-threshold}. These probabilities describe a distribution across sectors. Within a sector, the curvature mass characterizes the quadratic excitation, and with dynamical gravitons its interacting response also includes the two-graviton continuum.

Unknown wormhole determinants prevent a unique conversion of the effective mass scale into a Peccei--Quinn scale. The radial-freezing conditions in Eq.~\eqref{eq:throat-mass-ratios} are not guaranteed in a weakly coupled charge-one ultraviolet model, and a numerical value inferred from the frozen-modulus action should not be used outside that domain. Equation~\eqref{eq:two-harmonic-conditional-certificate} shows why the small unconditional error weight gives little information about the sectors selected by a very small unit-harmonic mass.

The threshold coefficient and the crossover are predictions for the specified two-channel state and truncation. Additional harmonics or a different state measure require a new global-minimum calculation.

\appendix

\section{Euclidean scalar--two-form duality and the flux sum}
\label{app:duality}

For the integer-normalized three-form in Eq.~\eqref{eq:flux-quantization}, a convenient first-order Euclidean functional is
\begin{equation}
 S_1[H,\theta]
 =
 \frac{1}{2\mu^2}\int H\wedge\star H
 +\ii\int \dd\theta\wedge H.
 \label{eq:first-order-duality}
\end{equation}
The compact field $\theta\sim\theta+2\pi$ enforces both the local Bianchi identity and the global sum over integral fluxes.  Varying the noncompact part of $H$ gives
\begin{equation}
 H=-\ii\mu^2\star\dd\theta.
 \label{eq:euclidean-duality-relation}
\end{equation}
Substitution into Eq.~\eqref{eq:first-order-duality}, using $\star^2=-1$ on three-forms in four Euclidean dimensions, gives the positive scalar action
\begin{equation}
 S_E[\theta]
 =
 \frac{\mu^2}{2}\int\dd\theta\wedge\star\dd\theta.
 \label{eq:scalar-euclidean-action}
\end{equation}
The factor of $\ii$ in Eq.~\eqref{eq:euclidean-duality-relation} is the local indication that a real two-form saddle is not mapped to a real semiclassical scalar configuration.

The global equivalence is a Poisson resummation.  On a fixed geometry, let $k\in\mathbb Z$ label scalar winding between the two ends and let
$\eta=(\theta_+-\theta_-)/(2\pi)$.  The relevant theta-function identity is
\begin{equation}
 \sum_{k\in\mathbb Z}
 \ee^{-A(k+\eta)^2}
 =
 \sqrt{\frac{\pi}{A}}
 \sum_{n\in\mathbb Z}
 \ee^{-\pi^2n^2/A}
 \ee^{\ii n(\theta_+-\theta_-)}.
 \label{eq:poisson-resummation}
\end{equation}
The Fourier label $n$ is the quantized three-form flux.  Equation~\eqref{eq:poisson-resummation} shows directly that a charge-$n$ throat contributes
$\ee^{\ii n\theta_+}\ee^{-\ii n\theta_-}$ at long distances.  After integrating over the two mouth positions this is Eq.~\eqref{eq:general-bilocal}.  The Poisson resummation, rather than a pointwise classical duality, is therefore what produces the scalar mouth operators \cite{Witten2026}.

With the $2\pi$-normalized field strength
$H^{(2\pi)}=2\pi H$, Eq.~\eqref{eq:first-order-duality} becomes
\begin{equation}
 S_1
 =
 \frac{1}{2(2\pi)^2\mu^2}
 \int H^{(2\pi)}\wedge\star H^{(2\pi)}
 +\frac{\ii}{2\pi}\int\dd\theta\wedge H^{(2\pi)},
 \label{eq:two-pi-normalization}
\end{equation}
with $\int H^{(2\pi)}=2\pi n$.  This provides a direct conversion between common normalization conventions.

\section{On-shell action and boundary prescriptions}
\label{app:action}

For Eq.~\eqref{eq:o4-ansatz},
\begin{equation}
 H_{\mu\nu\rho}H^{\mu\nu\rho}
 =
 \frac{3n^2}{2\pi^4a^6}.
 \label{eq:h-squared}
\end{equation}
Using Eqs.~\eqref{eq:friedmann-wormhole} and \eqref{eq:wormhole-curvature}, the gravitational and matter bulk densities are equal on shell,
\begin{equation}
 -\frac{\mpl^2}{16\pi}R
 =
 \frac{1}{2\cdot3!\mu^2}H^2
 =
 \frac{3\mpl^2L_n^4}{8\pi a^6}.
 \label{eq:on-shell-density-equality}
\end{equation}
The full two-ended bulk action is therefore
\begin{align}
 S_{n,{\rm bulk}}^{(2)}
 &=
 \frac{3\pi\mpl^2L_n^4}{2}
 \int_{-\infty}^{\infty}\frac{\dd\tau}{a^3}
 \nonumber\\
 &=
 \frac{3\pi\mpl^2L_n^4}{2}
 \left[
 2\int_{L_n}^{\infty}
 \frac{\dd a}{a^3\sqrt{1-L_n^4/a^4}}
 \right].
 \label{eq:bulk-action-integral}
\end{align}
The radial integral is
\begin{equation}
 2\int_{L_n}^{\infty}
 \frac{\dd a}{a^3\sqrt{1-L_n^4/a^4}}
 =\frac{\pi}{2L_n^2},
 \label{eq:radial-integral}
\end{equation}
which gives Eq.~\eqref{eq:connected-action}.  For a large-radius boundary, the difference between the extrinsic curvature of the wormhole end and that of a flat reference sphere is $\cO(a^{-5})$.  Multiplication by the boundary volume $\cO(a^3)$ therefore gives a vanishing background-subtracted Gibbons--Hawking--York contribution at either asymptotic end.

A half geometry has bulk action
\begin{equation}
 S_{n,{\rm bulk}}^{(1)}
 =
 \frac{3\pi^2}{8}\mpl^2L_n^2
 =
 \frac{\sqrt{3\pi}}{8}\frac{|n|\mpl}{\mu}.
 \label{eq:half-bulk-action}
\end{equation}
If the throat is treated as an actual boundary, the chosen variational problem requires an internal Gibbons--Hawking--York term.  With the flat reference subtraction used in Ref.~\cite{AlonsoUrbano2019},
\begin{equation}
 S_{{\rm GHY},{\rm throat}}
 =
 -\frac{3\pi}{4}\mpl^2L_n^2,
 \label{eq:throat-ghy}
\end{equation}
so that
\begin{align}
 S_{n,{\rm cut}}^{(0)}
 &=S_{n,{\rm bulk}}^{(1)}+S_{{\rm GHY},{\rm throat}}
 \nonumber\\
 &=
 \frac{3\pi^2}{8}\mpl^2L_n^2
 \left(1-\frac{2}{\pi}\right),
 \label{eq:cut-action-derived}
\end{align}
which is Eq.~\eqref{eq:cut-action}.  The smooth connected manifold has no boundary at the throat, so Eq.~\eqref{eq:throat-ghy} is not included in Eq.~\eqref{eq:connected-action}.  This is why the connected action cannot be obtained by simply doubling the cut action.

\section{Gaussian pushforward and exact propagator}
\label{app:gaussian}

For Eq.~\eqref{eq:gaussian-alpha}, polar integration gives
\begin{equation}
 \int\dd\nu_\sigma F(|\alpha|)
 =
 \int_0^\infty\dd r\,
 \frac{2r}{\sigma^2}\ee^{-r^2/\sigma^2}F(r).
 \label{eq:gaussian-polar-integral}
\end{equation}
With $x=\gamma r$, the Jacobian $\dd r=\dd x/\gamma$ yields Eq.~\eqref{eq:gaussian-mass-density}.  The moments follow from
\begin{equation}
 \int_0^\infty\dd u\,2u^{k+1}\ee^{-u^2}
 =\Gamma\left(1+\frac{k}{2}\right).
 \label{eq:moment-integral}
\end{equation}

For the propagator, set $x=m_\sigma^2u$ and $s=m_\sigma^2z$.  Then
\begin{equation}
 m_\sigma^2G_{{\rm ens},c}(s)
 =
 2\int_0^\infty\dd u\,
 \frac{u\ee^{-u^2}}{u+z}.
 \label{eq:dimensionless-propagator-integral}
\end{equation}
Using $u/(u+z)=1-z/(u+z)$ and, for $z>0$,
\begin{equation}
 \int_0^\infty\frac{\ee^{-u^2}}{u+z}\dd u
 =
 \frac{\ee^{-z^2}}{2}
 \left[
 \pi\operatorname{erfi}(z)-\operatorname{Ei}(z^2)
 \right],
 \label{eq:auxiliary-integral}
\end{equation}
one obtains Eq.~\eqref{eq:gaussian-propagator-exact}.  The expansions
\begin{align}
 \operatorname{Ei}(z^2)
 &=\gamma_E+2\ln z+z^2+\cO(z^4),
 \nonumber\\
 \operatorname{erfi}(z)
 &=\frac{2z}{\sqrt\pi}+\cO(z^3).
 \label{eq:special-function-expansions}
\end{align}
give Eq.~\eqref{eq:gaussian-propagator-small-s}.  At large Euclidean momentum, the ordinary moment expansion is recovered,
\begin{equation}
 G_{{\rm ens},c}(s)
 =
 \frac{1}{s}
 -\frac{\langle m_\phi^2\rangle}{s^2}
 +\cO(s^{-3}).
 \label{eq:large-momentum-expansion}
\end{equation}

\section{Oscillator states and the noncentral mass distribution}
\label{app:state-spectrum}

Write $a_\epsilon=(Q_\epsilon+\ii P_\epsilon)/\sqrt2$, with $[Q_\epsilon,P_{\epsilon'}]=\ii\delta_{\epsilon\epsilon'}$. Equation~\eqref{eq:bu-oscillator-pair} then gives
\begin{equation}
 \widehat X=(Q_++Q_-)/\sqrt2,
 \qquad \widehat Y=(P_+-P_-)/\sqrt2.
 \label{eq:commuting-epr-quadratures}
\end{equation}
Both vacuum variances are $1/2$, their covariance is zero, and their commutator vanishes. The vacuum expectation of the joint Weyl operator is consequently $\exp[-(k_c^2+k_s^2)/4]$. The transformations
\begin{equation}
 S_2^\dagger\widehat C S_2=\ee^{-\zeta}\widehat C,
 \qquad
 (D_+D_-)^\dagger\widehat C(D_+D_-)=\widehat C+\alpha_0
 \label{eq:bu-state-transformations}
\end{equation}
establish Eq.~\eqref{eq:bu-characteristic}. Fourier inversion gives Eq.~\eqref{eq:displaced-gaussian-state}.

Rotate $\alpha_0$ to the positive real axis and write $\alpha=r\ee^{\ii\vartheta}$. The identity
\begin{equation}
 I_0(z)=\frac{1}{2\pi}\int_0^{2\pi}\ee^{z\cos\vartheta}\dd\vartheta
 \label{eq:bessel-angular-identity}
\end{equation}
yields the radial density $2r\sigma^{-2}\exp[-r^2/\sigma^2-q^2]I_0(2qr/\sigma)$. The change of variable in Eq.~\eqref{eq:unit-mass} gives Eq.~\eqref{eq:noncentral-mass-density}. Its dimensionless moments follow by expanding $I_0$ in its power series:
\begin{equation}
 \langle u^k\rangle
 =\Gamma(1+k/2)\,{}_1F_1(-k/2;1;-q^2),
 \qquad k>-2,
 \label{eq:noncentral-moments}
\end{equation}
where ${}_1F_1$ is the confluent hypergeometric function. Setting $k=1$ and $k=2$ gives Eq.~\eqref{eq:noncentral-mean} and $\langle u^2\rangle=1+q^2$, respectively. For unit residue the zero-momentum averaged propagator is also finite:
\begin{equation}
 G_{{\rm ens},c}^{(q)}(0)
 =\frac{\sqrt\pi}{m_\sigma^2}
 \ee^{-q^2/2}I_0(q^2/2).
 \label{eq:noncentral-zero-propagator}
\end{equation}
Its first nonanalytic correction is fixed by Eq.~\eqref{eq:noncentral-threshold} and Eq.~\eqref{eq:linear-threshold-log}. No isolated massless pole follows from a zero lower threshold.
\section{Curvature and probability bounds for two harmonics}
\label{app:two-harmonic-bound}

For Eq.~\eqref{eq:two-harmonic-audit-potential}, global minimality implies
$V(\chi_\star+\pi)-V(\chi_\star)=2A\cos\chi_\star\geq0$. Set
$c=\cos\chi_\star\geq0$ and $w=\cos(2\chi_\star+\delta)$. Stationarity and the curvature are
\begin{equation}
 A\sin\chi_\star=-2B\sin(2\chi_\star+\delta),
 \qquad k_\star=Ac+4Bw.
 \label{eq:two-harmonic-stationarity-audit}
\end{equation}
The upper bound in Eq.~\eqref{eq:two-harmonic-uniform-bound} follows immediately. When $A\geq2B$, stationarity yields the identity
\begin{equation}
 \begin{split}
 k_\star-A+4B
 &=4B(1+w)-A(1-c)\\
 &=4B(1+w)
 \left[1-\frac{B(1-w)}{A(1+c)}\right]\geq0,
 \end{split}
 \label{eq:two-harmonic-lower-proof}
\end{equation}
since $A(1-c)=4B^2(1-w^2)/[A(1+c)]$, $1+c\geq1$, and $1-w\leq2$. When $A<2B$, stability gives $k_\star\geq0>A-4B$. This proves the stated bound, including cases with degenerate global minima. At $A=0$ and $B>0$ the curvature is obtained directly from the second harmonic.

For Eq.~\eqref{eq:second-harmonic-uncertified-weight}, let $R_i=|\alpha_i|/\sigma_i$ and $t=\epsilon^2/(16\kappa_h^2)$. The centered second channel has $\Pr(R_2^2>z)=\ee^{-z}$. Conditional on $R_1$, the uncertified probability is therefore $\ee^{-tR_1^2}$. The first-channel Gaussian integral gives
\begin{equation}
 \left\langle\ee^{-tR_1^2}\right\rangle
 =\frac{1}{1+t}\exp\!\left[-\frac{q^2t}{1+t}\right],
 \label{eq:noncentral-radius-laplace}
\end{equation}
which proves Eq.~\eqref{eq:second-harmonic-uncertified-weight}. For $q=0$ it reduces to the ratio law for two independent unit exponential variables.

For the change of variables used in Eq.~\eqref{eq:two-harmonic-exact-density}, the stationarity combination is $\operatorname{Im}(c_1\ee^{\ii\theta}+2c_2\ee^{2\ii\theta})$. Its derivative with respect to $\theta$ at $\theta_\star$ is $y$. In the rotated coordinates, the three tangent columns for $(y,U,W)$ have determinant one in the remaining coordinates $(\operatorname{Re}(c_1\ee^{\ii\theta_\star}),\operatorname{Re}(c_2\ee^{2\ii\theta_\star}),\operatorname{Im}(c_2\ee^{2\ii\theta_\star}))$. Rotation has unit determinant, which gives the Jacobian $y$. For $y>0$, equality in the global condition traces a surface of zero measure in the coefficient space. For $y=0$ the condition requires $W=0$ and $U\leq0$, also a set of zero measure after mapping back to the coefficients.

For the threshold limit, at bounded small $y$ the Gaussian--Bessel factor $\ee^{-L^2}I_0(2qL)$ is bounded by $\ee^{q^2}$. After $W=\sqrt{y}\,w$, the domain satisfies $|w|\leq\sqrt{y/4-U}$. The remaining Gaussian in $U$ gives an integrable majorant, so the limit can be taken inside the integral. At $U=-v<0$ the limiting $w$ interval has length $2\sqrt v$. This gives the leading density proportional to $y^{3/2}$ and the coefficient in Eq.~\eqref{eq:two-harmonic-tail-coefficient} upon integration over $y$. The same argument for Eq.~\eqref{eq:two-harmonic-crossover} omits the first-channel Gaussian--Bessel factor. At large $z$, the remainder in that equation is bounded by $\mathbb E[(16R^2-z^2)_+]=16\ee^{-z^2/16}$. Equation~\eqref{eq:two-harmonic-propagator-threshold} follows by subtracting $1/x-s/x^2$ from $1/(x+s)$ in the Stieltjes transform. The remaining integral has a threshold contribution $D_{12}s^{3/2}\int_0^\infty t^{-1/2}(1+t)^{-1}\dd t$.

The actual relative error can be integrated from the same coefficient measure. Write $(U,W)=r(\cos\psi,\sin\psi)$ and $y=r\ell$, so $r=B/a_\sigma$. Then $\ell\geq2(1+\cos\psi)$, $L=rd$, and
\begin{equation}
 d=\sqrt{(\ell-4\cos\psi)^2+4\sin^2\psi},\qquad
 \mathcal D_\epsilon(\psi)
 =\{\ell\geq2(1+\cos\psi):\ d>0,\ |\ell/d-1|>\epsilon\}.
 \label{eq:two-harmonic-error-domain}
\end{equation}
Integrating over $r$ in Eq.~\eqref{eq:two-harmonic-exact-density} gives
\begin{equation}
 \Pr\!\left(\frac{|m_{12}^2-m_1^2|}{m_1^2}>\epsilon\right)
 =\frac{\kappa_h^2\ee^{-q^2}}{\pi}
 \int_0^{2\pi}\!\dd\psi\int_{\mathcal D_\epsilon(\psi)}\!\dd\ell\,
 \frac{\ell(1+\tau)\ee^\tau}{(1+\kappa_h^2 d^2)^2},
 \qquad
 \tau=\frac{q^2\kappa_h^2d^2}{1+\kappa_h^2d^2}.
 \label{eq:two-harmonic-actual-error}
\end{equation}
The radial integral uses $\int_0^\infty r\ee^{-Dr^2} I_0(2qdr)\dd r=\ee^{q^2d^2/D}/(2D)$ for $D=d^2+\kappa_h^{-2}$, differentiated with respect to $D$. The domain is bounded for any $\epsilon>0$ by Eq.~\eqref{eq:two-harmonic-uniform-bound}, which implies $d<4/\epsilon$ and $\ell\leq d+4$ on the error event.

\section*{Statements and Declarations}

\noindent Funding: The author received no financial support for this work.

\noindent Competing interests: The author declares no competing interests.

\noindent Data availability: This theoretical study uses no external datasets. The results are derived from the equations presented in the article.

\bibliographystyle{elsarticle-num}
\bibliography{biblio}
\end{document}